\documentclass{article}

     \usepackage[preprint]{neurips_2019}

\usepackage[utf8]{inputenc} 
\usepackage[T1]{fontenc}    
\usepackage{hyperref}       
\usepackage{url}            
\usepackage{booktabs}       
\usepackage{amssymb,amsmath,amsthm,bbm}
\usepackage{mathtools,enumitem}
\usepackage{amsfonts}       
\usepackage{nicefrac}       
\usepackage{microtype}      

\usepackage{floatrow}
\newtheorem{theorem}{Theorem}

\newtheorem{corollary}{Corollary}

\theoremstyle{definition}
\newtheorem{remark}{Remark}

\newcommand{\minimize}{\mathop{\mathrm{minimize}}}
\newcommand{\st}{\mathop{\mathrm{subject\,\,to}}}

\def\R{\mathbb{R}}
\def\E{\mathbb{E}}

\def\tr{\mathrm{tr}}

\def\hx{\hat{x}}
\def\bx{\bar{x}}
\def\bP{\bar{P}}
\def\hb{\hat{b}}
\def\hB{\hat{B}}
\def\hR{\hat{R}}

\newcommand{\sfrac}[2]{\,^#1\!/_#2}
\usepackage[singlelinecheck=false]{caption}

\title{Kalman Filter, Sensor Fusion, and Constrained Regression: Equivalences
  and Insights} 

\author{
	Maria Jahja \\
	Department of Statistics \\
	Carnegie Mellon University \\
	Pittsburgh, PA 15213 \\
	{\tt maria@stat.cmu.edu} \\
	\And
	David Farrow \\
	Computational Biology Department\\
	Carnegie Mellon University \\
	Pittsburgh, PA 15213 \\
	{\tt dfarrow0@gmail.com} \\
	\AND
	Roni Rosenfeld \\
	Machine Learning Department \\
	Carnegie Mellon University \\
	Pittsburgh, PA 15213 \\
	{\tt roni@cs.cmu.edu} \\
	\And
	Ryan J. Tibshirani \\
	Department of Statistics \\
	Machine Learning Department \\
	Carnegie Mellon University \\
	Pittsburgh, PA 15213 \\
	{\tt ryantibs@stat.cmu.edu} 
}
\begin{document}
	\maketitle
	
	\begin{abstract}
		The Kalman filter (KF) is one of the most widely used tools for data
		assimilation and sequential estimation. In this work, we show that the state
		estimates from the KF in a standard linear dynamical system setting are
		equivalent to those given by the KF in a transformed system, with infinite
		process noise (i.e., a ``flat prior'') and an augmented measurement space. This
		reformulation---which we refer to as augmented measurement sensor fusion
		(SF)---is conceptually interesting, because the transformed system here is
		seemingly static (as there is effectively no process model), but we can still
		capture the state dynamics inherent to the KF by folding the process model into
		the measurement space.  Further, this reformulation of the KF turns out to be
		useful in settings in which past states are observed eventually (at some lag).
		Here, when the measurement noise covariance is estimated by the empirical
		covariance, we show that the state predictions from SF are equivalent to those
		from a regression of past states on past measurements, subject to particular
		linear constraints (reflecting the relationships encoded in the measurement
		map). This allows us to port standard ideas (say, regularization methods) in
		regression over to dynamical systems. For example, we can posit multiple
		candidate process models, fold all of them into the measurement model, transform
		to the regression perspective, and apply $\ell_1$ penalization to perform
		process model selection. We give various empirical demonstrations, and focus on
		an application to nowcasting the weekly incidence of influenza in the US.
	\end{abstract}
	
	\section{Introduction}
	\label{sec:intro}
	
	Let $x_t \in \R^k$, $t=1,2,3,\ldots$ denote states and $z_t \in \R^d$,
	$t=1,2,3,\ldots$ denote measurements evolving according to the time-invariant
	linear dynamical system:
	\begin{align}
		\label{eq:process} 
		x_t &= F x_{t-1} + \delta_t, \\
		\label{eq:measure} 
		z_t &= H x_t + \epsilon_t,
	\end{align} 
	for $t=1,2,3,\ldots$.  We assume the noise terms $\delta_t,\epsilon_t$ have mean
	zero and covariances $Q \in \R^{k \times k}$ and $R \in \R^{d \times d}$,
	respectively, for all $t=1,2,3,\ldots$.  Also, we assume that the initial state
	$x_0$ and all noise terms are mutually independent. We call \eqref{eq:process}
	the process model and \eqref{eq:measure} the measurement model.
	
	\paragraph{Kalman filter.} The Kalman filter (KF) \citep{kalman1960new} is a    
	method for sequential estimation in the model \eqref{eq:process},
	\eqref{eq:measure}. Given past estimates \smash{$\hx_1,\ldots,\hx_t$} and
	measurements $z_1,\ldots,z_{t+1}$, we form an estimate \smash{$\hx_{t+1}$} of
	the state $x_{t+1}$ via 
	\begin{align}
		\label{eq:kf_predict}
		\bx_{t+1} &= F \hx_t, \\
		\label{eq:kf_update}
		\hx_{t+1} &= \bx_{t+1} + K_{t+1} (z_{t+1} - H \bx_{t+1}),
	\end{align}
	where $K_{t+1} \in \R^{k \times d}$ is called the {\it Kalman gain} (at
	time $t+1$).  It is itself updated sequentially, via
	\begin{align}
		\label{eq:p_predict}
		\bP_{t+1} &= F P_t F^T + Q, \\
		\label{eq:k_update}
		K_{t+1} &= \bP_{t+1} H^T (H \bP_{t+1} H^T + R)^{-1}, \\
		\label{eq:p_update}
		P_{t+1} &= (I-K_{t+1} H) \bP_{t+1}.
	\end{align}
	where $P_{t+1} \in \R^{k \times k}$ denotes the state error covariance (at time
	$t+1$).  The step \eqref{eq:kf_predict} is often called the {\it predict} step:
	we form an intermediate estimate \smash{$\bx_{t+1}$} of the state based on
	the process model and our estimate at the previous time point.  The
	step \eqref{eq:kf_update} is often called the {\it update} step: we update our
	estimate \smash{$\hx_{t+1}$} based on the measurement model and the
	measurement $z_{t+1}$.
	
	Under the data model \eqref{eq:process}, \eqref{eq:measure} and the conditions
	on the noise stated above, the Kalman filter attains the optimal mean squared
	error $\E \|\hx_t - x_t\|_2^2$ among all linear unbiased filters, at each
	$t=1,2,3,\ldots$. When the initial state $x_0$ and all noise terms are Gaussian,
	the Kalman filter estimates exactly reduce to the Bayes estimates
	\smash{$\hx_t=\E(x_t|z_1,\ldots,z_t)$}, $t=1,2,3,\ldots$. Numerous important
	extensions have been proposed, e.g., the ensemble Kalman filter (EnKF)
	\citep{evensen1994sequential,houtekamer1998data}, which approximates the noise
	process covariance $Q$ by a sample covariance in an ensemble of state
	predictions, as well as the extended Kalman filter (EKF)
	\citep{smith1962application} and unscented Kalman filter (UKF)
	\citep{julier1997new}, which both allow for nonlinearities in the process
	model. Particle filtering (PF) \citep{gordon1993novel} has more recently become
	a popular approach for modeling complex dynamics. PF adaptively approximates 
	the posterior distribution, and in doing so, avoids the linear and Gaussian
	assumptions inherent to the KF. This flexibility comes at the cost of a greater
	computational burden.
	
	In this paper, we revisit the standard KF \eqref{eq:kf_predict},
	\eqref{eq:kf_update} and show that its estimates $\hx_{t+1}$, $t=0,1,2,\ldots$  
	are equivalent to those from the KF applied to a transformed system, with 
	infinite process noise and an augmented measurement space.  At first glance,
	this is perhaps surprising, because the transformed system effectively lacks a
	process model and is therefore seemingly static; however, it is able to take the state 
	dynamics into account as part of its measurement model.  Importantly, this
	reformulation of the KF leads us to derive a second, key reformulation for
	problems in which past states are observed (at some lag).  This second
	reformulation is the methodological crux of our paper: it is a constrained
	regression approach for predicting states from measurements, motivated by 
	(derived from) SF and the KF.  We illustrate its effectiveness in an application
	to nowcasting weekly influenza levels in the US. 
	
	\paragraph{Sensor fusion.} If we let the noise covariance in the process model   
	diverge to infinity, $Q \to \infty$\footnote{To make this unambiguous, we may
		take, say, $Q = a I$ and let $a \to \infty$.}, 
	then the Kalman filter estimate in \eqref{eq:kf_predict}, \eqref{eq:kf_update} 
	simplifies to  
	\begin{equation}
		\label{eq:sf}
		\hx_{t+1} = (H^T R^{-1} H)^{-1} H^T R^{-1} z_{t+1}.
	\end{equation}
	This can be verified by rewriting the Kalman gain as 
	\smash{$K_{t+1}=(\bP_{t+1}^{-1} + H^T R^{-1} H)^{-1} H^T R^{-1}$}, and observing
	that \smash{$\bP_{t+1}^{-1} \to 0$} as $Q \to \infty$.
	Alternatively, we can verify this by specializing to the case of Gaussian noise:
	as $\tr(Q) \to \infty$, we approach a flat prior, and the Kalman filter (Bayes 
	estimator) just maximizes the likelihood of $z_{t+1}|x_{t+1}$.  From the
	measurement model \eqref{eq:measure} (assuming Gaussian noise), this is a 
	weighted regression of $z_{t+1}$ on the measurement map $H$, precisely as in 
	\eqref{eq:sf}.
	
	We will call \eqref{eq:sf} the {\it sensor fusion} (SF) estimate (at time
	$t+1$).\footnote{ ``Sensor fusion'' is typically used as a generic term, 
		similar to ``data assimilation''; we use it to specifically describe the
		estimate in \eqref{eq:sf} to distinguish it from the KF.  This is useful when
		we describe equivalences, shortly.}   
	In this setting, we will also refer to the measurements as {\it sensors}. As
	defined, sensor fusion is a special case of the Kalman filter when there is
	infinite process noise; said differently, it is a special case of the Kalman
	filter when there is no process model at all.  Thus, looking at \eqref{eq:sf},
	the state dynamics have apparently been completely lost. Perhaps surprisingly, 
	as we will show shortly, these dynamics can be exactly recovered by augmenting
	the measurement vector $z_{t+1}$ with the KF intermediate prediction 
	\smash{$\bx_{t+1}=F \hx_t$} in \eqref{eq:kf_predict} (and adjusting the map $H$
	and covariance $R$ appropriately).  We summarize this and our other
	contributions next. 
	
	\paragraph{Summary of contributions.} An outline of our contributions in this 
	paper is as follows. 
	
	\begin{enumerate}
		\item We show in Section \ref{sec:kf_sf} that, if we take the KF intermediate
		prediction \smash{$\bx_{t+1}$} in \eqref{eq:kf_predict}, append it to the
		measurement vector $z_{t+1}$, and perform SF \eqref{eq:sf} (with an
		appropriately adjusted $H,R$), then the result is exactly the KF estimate
		\eqref{eq:kf_update}.   
		
		\item We show in Section \ref{sec:sf_reg} that, if we are in a problem setting
		in which past states are observed (at some lag, which is the case in the flu
		nowcasting application), and we replace the noise covariance $R$ from the
		measurement model by the empirical covariance on past data, then the sensor
		fusion estimate \eqref{eq:sf} can be written as \smash{$\hB^T z_{t+1}$}, where 
		\smash{$\hB \in \R^{d \times k}$} is a matrix of coefficients that solves
		a regression problem of the states on the measurements (using past data), 
		subject to the equality constraint \smash{$H^T \hB = I$}.  
		
		\item We demonstrate the effectiveness of our new regression formulation of SF
		in Section \ref{sec:flu_demo} by describing an application of this methodology
		to nowcasting the incidence of weekly flu in the US.  This achieves
		state-of-the art performance in this problem.    
		
		\item We present in Section \ref{sec:end} some extensions of the regression
		formulation of SF; they do not have direct equivalences to SF (or the KF), 
		but are intuitive and extend dynamical systems modeling in new directions
		(e.g., using $\ell_1$ penalization to perform a kind of process model
		selection).  
	\end{enumerate}
	
	We make several remarks.  The equivalences described in points 1--3 above are  
	deterministic (they do not require the modeling assumptions \eqref{eq:process},
	\eqref{eq:measure}, or any modeling assumptions whatsoever).  Further, even
	though their proofs are elementary (they are purely linear algebraic) and the
	setting is a classical one (linear dynamical systems), these equivalences
	are---as far as we can tell---new results.  They deserve to be widely known and 
	may have implications beyond what is explored in this paper.   
	
	For example, the regression formulation of SF may still be a useful perspective
	for problems in which past states are fully unobserved (this being the case in
	most KF applications). In such problems, we may consider using {\it smoothed} 
	estimates of past states, obtained by running a backward version of the KF
	forward recursions \eqref{eq:kf_predict}--\eqref{eq:p_update} (see, e.g.,
	Chapter 7 of \citet{anderson1979optimal}), for the purposes of the regression
	formulation. As another example, the SF view of the KF may be a useful
	formulation for the purposes of estimating the covariances $R,Q$, or the maps
	$F,H$, or all of them; in this paper, we assume that $F,H,R,Q$ are known (except
	for in the regression formulation of SF, in which $R$ is unknown but past
	states are available); in general, there are well-developed methods for
	estimating $F,H,R,Q$ such as {\it subspace identification} algorithms (see,
	e.g., \citet{vanovershee1996subspace}), and it may be interesting to see if the
	SF perspective offers any advantages here.  
	
	\paragraph{Related work.} The Kalman filter and its extensions, as previously 
	referenced (EnKF, EKF, UKF), are the de facto standard in state    
	estimation and tracking problems; the literature surrounding them is
	enormous and we cannot give a thorough treatment.  Various authors have pointed
	out the simple fact that maximum likelihood estimate in \eqref{eq:sf}, which
	we call sensor fusion, is the limit of the KF as the noise covariance in the
	process model approaches infinity (see, e.g., Chapter 5.9 of
	\citet{brown2012introduction}).  We have not, however, seen any authors
	note that this static model can recover the KF by augmenting the measurement
	vector with the KF intermediate prediction (Theorem \ref{thm:kf_sf}).  
	
	Along the lines of our second equivalence (Theorem \ref{thm:sf_reg}), there is
	older work in the statistical calibration literature that studies the
	relationships between the regressions of $y$ on $x$ and $x$ on $y$ (for
	multivariate $x,y$, see \citet{brown1982multivariate}).  This is somewhat
	related to our result, since we show that a {\it backwards} or {\it indirect} 
	approach, which models $z_{t+1}|x_{t+1}$, is actually equivalent to a {\it
		forwards} or {\it direct} approach, which predicts $x_{t+1}$ from $z_{t+1}$
	via regression.  However, the details are quite different.
	
	Finally, our SF methodology in the flu nowcasting application blends together
	individual predictors in a way that resembles {\it linear stacking} 
	\citep{wolpert1992stacked,breiman1996stacked}.  In fact, one implication of our
	choice of measurement map $H$ in the flu nowcasting problem, as well as the
	constraints in our regression formulation of SF, is that all regression weights
	must sum to 1, which is the standard in linear stacking as well.  However, the 
	equality constraints in our regression formulation are quite a bit more complex,
	and reflect aspects of the sensor hierarchy that linear stacking would not.        
	
	\section{Equivalence between KF and SF}
	\label{sec:kf_sf}
	
	As already discussed, the sensor fusion estimate \eqref{eq:sf} is a limiting 
	case of the Kalman filter \eqref{eq:kf_predict}, \eqref{eq:kf_update}, and  
	initially, it seems, one rather limited in scope: there is effectively no
	process model (as we have sent the process variance to infinity). However, 
	as we show next, the KF is actually itself a special case of SF, when we augment
	the measurement vector by the KF intermediate predictions, and appropriately
	adjust the measurement map $H$ and noise covariance $R$.  The proof is
	elementary, a consequence of the Woodbury matrix and related manipulations. It
	is given in the supplement.   
	
	\begin{theorem}
		\label{thm:kf_sf}
		At each time $t=0,1,2,\ldots$, suppose we augment our measurement vector by
		defining \smash{$\tilde{z}_{t+1} = (z_{t+1}, \bx_{t+1}) \in \R^{d+k}$}, where 
		\smash{$\bx_{t+1}=F \hx_t$} is the KF intermediate prediction at time
		$t+1$. Suppose that we also augment our measurement map by defining
		\smash{$\tilde{H} \in \R^{(d+k) \times k}$} to be the rowwise concatention of
		$H$ and the identity matrix $I \in \R^{k \times k}$.  Furthermore, suppose we
		define an augmented measurement noise covariance   
		\begin{equation}
			\label{eq:block_cov}
			\tilde{R}_{t+1} = \begin{bmatrix} R & 0 \\ 0 & \bP_{t+1} \end{bmatrix}, 
		\end{equation}
		where \smash{$\bP_{t+1}$} is the KF intermediate error covariance at time $t+1$
		(as in \eqref{eq:p_predict}). Then applying SF to the augmented system produces
		an estimate at $t+1$ that equals the KF estimate,
		\begin{equation}
			\label{eq:kf_sf}
			(\tilde{H}^T \tilde{R}_{t+1}^{-1} \tilde{H})^{-1} \tilde{H}^T
			\tilde{R}_{t+1}^{-1} \tilde{z}_{t+1} = 
			\bx_{t+1} + K_{t+1} (z_{t+1} - H \bx_{t+1}), 
		\end{equation}
		where $K_{t+1}$ is the Kalman gain at $t+1$ (as in \eqref{eq:k_update}).
	\end{theorem}
	
	\begin{remark}
		We can think of the last state estimate \smash{$\hx_t$} in the theorem (which is 
		propagated forward via \smash{$\bx_{t+1}=F \hx_t$}) as the previous output
		from SF itself, when applied to the appropriate augmented system. More
		precisely, by induction, Theorem \ref{thm:kf_sf} says that iteratively applying
		SF to \smash{$\tilde{z}_{t+1},\tilde{H},\tilde{R}_{t+1}$} across times
		$t=0,1,2,\ldots$, where each \smash{$\bx_{t+1}=F \hx_t$} is the intermediate 
		prediction using the last SF estimate \smash{$\hx_t$}, produces a sequence
		\smash{$\hx_{t+1}$}, $t=0,1,2,\ldots$ that matches the state estimates from the
		KF.   
	\end{remark}
	
	\begin{remark}
		The result in Theorem \ref{thm:kf_sf} can be seen from a Bayesian
		perspective, as was pointed out by an anonymous reviewer. When the initial state
		$x_0$ and all noise terms in \eqref{eq:process}, \eqref{eq:measure} are
		Gaussian, recall the KF reduces to the Bayes estimator. Here the posterior is
		the product of a Gaussian likelihood and Gaussian prior, and is thus itself 
		Gaussian.  (The proof of this standard fact uses similar arguments to the proof
		of Theorem \ref{thm:kf_sf}.)  Meanwhile, in augmented SF, we can view the
		Gaussian likelihood being maximized as the product of the Gaussian density of
		$z_{t+1}$ and that of \smash{$\bx_{t+1}$}.  This matches the posterior used by
		the KF, where the density of \smash{$\bx_{t+1}$} plays the role of the prior in
		the KF.  Therefore in each case, we are defining our estimate to be the mean of
		the same Gaussian distribution.
	\end{remark}
	
	\begin{remark}
		The equivalence between SF and KF can be extended beyond the case of linear 
		process and linear measurement models. Given a nonlinear process map $f$ and 
		a nonlinear process model $h$, suppose we define \smash{$\bx_{t+1}=f(\hx_t)$},    
		\smash{$F_{t+1}=Df(\hx_t)$} (the Jacobian of $f$ at \smash{$\hx_t$}), and 
		\smash{$H_{t+1}=Dh(\bx_{t+1})$} (the Jacobian of $h$ at \smash{$\bx_{t+1}$}). 
		Suppose we define the augmented measurement vector as 
		\begin{equation}
			\label{eq:tz_nonlinear}
			\tilde{z}_{t+1}=\big(z_{t+1}+H_{t+1} \bx_{t+1} - h(\bx_{t+1}), \,
			\bx_{t+1}\big), 
		\end{equation}
		where we have offset the measurement $z_{t+1}$ by the residual 
		\smash{$H_{t+1} \bx_{t+1} - h(\bx_{t+1})$} from linearization.  Suppose, as in
		the theorem, we define the augmented measurement map  
		\smash{$\tilde{H}_{t+1} \in \R^{(d+k) \times k}$} to be the rowwise
		concatenation of $H_{t+1}$ and $I \in \R^{k \times k}$, and define
		\smash{$\tilde{R}_{t+1} \in \R^{(d+k) \times (d+k)}$} as in
		\eqref{eq:block_cov}, for \smash{$\bP_{t+1}$} as in \eqref{eq:p_predict}, but
		with $F_{t+1},H_{t+1}$ in place of $F,H$.  In the supplement, we prove that 
		\begin{equation}
			\label{eq:kf_sf_nonlinear}
			(\tilde{H}_{t+1}^T \tilde{R}_{t+1}^{-1} \tilde{H}_{t+1})^{-1} \tilde{H}_{t+1}^T  
			\tilde{R}_{t+1}^{-1} \tilde{z}_{t+1} = 
			\bx_{t+1} + K_{t+1} \big(z_{t+1} - h(\bx_{t+1})\big),
		\end{equation}
		where $K_{t+1}$ is as in \eqref{eq:k_update}, but with $F_{t+1},H_{t+1}$ in
		place of $F,H$.  The right-hand side above is precisely the {\it extended} KF 
		(EKF).  The left-hand side is what we might call {\it extended} SF (ESF).  
	\end{remark}
	
	
	\section{Equivalence between SF and regression}
	\label{sec:sf_reg}
	
	Suppose that in our linear dynamical system, at each time $t$, we observe the 
	measurement $z_t$, make a prediction \smash{$\hx_t$} for $x_t$, then
	later observe the state $x_t$ itself. (This setup indeed describes the influenza 
	nowcasting problem, a central motivating example that we will describe   
	shortly.) In such problems, we can estimate $R$ using the empirical
	covariance on past data.  When we plug this into \eqref{eq:sf}, it turns out
	SF reduces to a prediction from a constrained regression of past states on past
	measurements.  
	
	\subsection{Equivalent regression problem}
	
	In making a prediction at time $t+1$, we assume in this section that we observe
	past states. We may assume without a loss of generality that we observe the full
	past $x_i$, $i=1,\ldots,t$ (if this is not the case, and we observe only some
	subset of the past, then the only changes to make in what follows are
	notational). Assuming the measurement noise covariance $R$ is unknown, we may
	use 
	\begin{equation}
		\label{eq:emp_cov}
		\hR_{t+1} = \frac{1}{t} \sum_{i=1}^t (z_i - Hx_i)(z_i - Hx_i)^T, 
	\end{equation}
	the empirical (uncentered) covariance based on past data, as an estimate.  Under
	this choice, it turns out that sensor fusion \eqref{eq:sf} is exactly equivalent
	to a regression of states on measurements, subject to certain equality 
	constraints.  The proof is elementary, but requires detailed arguments.  It   
	is deferred until the supplement.
	
	\begin{theorem}
		\label{thm:sf_reg}
		Let \smash{$\hR_{t+1}$} be as in \eqref{eq:emp_cov} (assumed to be invertible). 
		Consider the SF prediction at time $t+1$, with \smash{$\hR_{t+1}$} in place of
		$R$. Denote this by \smash{$\hx_{t+1} = \hB^T z_{t+1}$}, where 
		$$
		\hB^T = (H^T \hR_{t+1}^{-1} H)^{-1} H^T \hR_{t+1}^{-1} 
		$$
		(and \smash{$H^T \hR_{t+1}^{-1} H$} is assumed invertible).  Each column of  
		\smash{$\hB$}, denoted \smash{$\hb_j \in \R^d$}, $j=1,\ldots,k$, solves
		\begin{equation}
			\label{eq:con_reg}
			\begin{alignedat}{2}
				&\minimize_{b_j \in \R^d} && \sum_{i=1}^t (x_{ij} - b_j^T z_i)^2 \\
				&\st \quad && H^T b_j = e_j,
			\end{alignedat}
		\end{equation}
		where $e_j \in \R^d$ is the $j$th standard basis vector (all 0s except for a 1
		in the $j$th component).
	\end{theorem}
	
	\begin{remark}
		As discussed in the introduction, the interpretation of \smash{$(H^T
			\hR_{t+1}^{-1} H)^{-1} H^T \hR_{t+1}^{-1} z_{t+1}$} as the coefficients from
		regressing $z_{t+1}$(the response) onto $H$ (the covariates) is more or less
		immediate. Interpreting the same quantity as \smash{$\hB^T z_{t+1} = (\hb_1^T 
			z_{t+1},\ldots,\hb_k^T z_{t+1})$}, the predictions from historically
		regressing $x_i$, $i=1,\ldots,t$ (the response) onto $z_i$, $i=1,\ldots,t$ (the
		covariates), however, is much less obvious.  The latter is a {\it forwards} or
		{\it direct} regression approach to predicting $x_{t+1}$, whereas SF was
		originally defined via the {\it backwards} or {\it indirect} perspective
		inherent to the measurement model \eqref{eq:measure}.
	\end{remark}
	
	
	\subsection{Influenza nowcasting}
	\label{sec:flu_setup}
	
	An example that we will revisit frequently, for the rest of the paper, is the
	following influenza (or flu) nowcasting problem.  The state variable of interest
	is the weekly percentage of weighted influenza-like illness (wILI), a measure of
	flu incidence provided by the Centers for Disease Control and Prevention (CDC),
	in each of the $k=51$ US states (including DC). Because it takes time for the
	CDC to collect and compile this data, they release wILI values with a 1 week
	delay. Meanwhile, various proxies for the flu (i.e., data sources that are
	potentially correlated with flu incidence) are available in real time, e.g., web
	search volume for flu-related terms, site traffic metrics for flu-related pages,
	pharmaceutical sales for flu-related products, etc.  We can hence train (using
	historical data) sensors to predict wILI, one from each data source, and plug
	them into sensor fusion \eqref{eq:sf} in order to ``nowcast'' the current flu
	incidence (that would otherwise remain unknown for another week).
	
	Such a sensor fusion system for flu nowcasting, using $ d = 308$ sensors (flu
	proxies), is described in Chapter 4 of 
	\citet{farrow2016modeling}\footnote{This is more than just a hypothetical   
		system; it is fully operational, and run by the Carnegie Mellon DELPHI group   
		to provide real-time nowcasts of flu incidence every week, in all US  
		states, plus select regions, cities, and territories. 
		(See \url{https://delphi.midas.cs.cmu.edu}).}.
	In addition to the surveillance sensors described above (search volume for flu
	terms, site traffic metrics for flu pages, etc.), the measurement vector in this
	nowcasting system also uses a sensor that is trained to make predictions of wILI
	using a seasonal autoregression with 3 lags (SAR3).  From the KF-SF equivalence 
	established in Section \ref{sec:kf_sf}, we can think of this SAR3 sensor as
	serving the role of something like a process model, in the underlying dynamical
	system.   
	
	While wILI itself is available at the US state level, the data source used to
	train each sensor may only be available at coarser geographic resolution.
	Thus, importantly, each sensor outputs a prediction at a different
	geographic resolution (which reflects the resolution of its corresponding data
	source).  As an example, the number of visits to flu-related CDC pages are
	available for each US state separately; so for each US state, we train a
	separate sensor to predict wILI from CDC site traffic. However, counts for
	Wikipedia page visits are only available nationally; so we train just one sensor
	to predict national wILI from Wikipedia page visits. 
	
	Assuming unbiasedness of all the sensors, we construct the map $H$ in
	\eqref{eq:measure} so that its rows reflect the geography of the sensors.  For 
	example, if a sensor is trained on data that is available at the $i$th US state,
	then its associated row in $H$ is  
	$$
	(0, \ldots \underset{\substack{\uparrow \\ i}}{1}, \ldots 0);
	$$
	and if a sensor is trained on data from the aggregate of the first 3 US states,
	then its associated row is  
	$$
	(w_1, w_2, w_3, 0, \ldots 0),
	$$
	for weights $w_1,w_2,w_3 > 0$ such that $w_1+w_2+w_3=1$, based on
	relative state populations; and so on. Figure \ref{fig:flu_simple} illustrates
	the setup in a simple example.  
	
	\begin{figure}[t]
		\floatbox[{\capbeside\thisfloatsetup{capbesideposition={right,center},
				capbesidewidth=0.45\textwidth}}]{figure}[\FBwidth]
		{\caption[short caption]{\it\small Simplified version of the flu nowcasting 
				problem, with $k=3$ states and $d=8$ sensors. We have a 3-level hierarchy,  
				where $x_1, x_2, x_3$ are part of the first region and $x_4, x_5$ are part of 
				the second. The national level is at the root. As for the sensors, we
				have one at each state, one at each region, and one at the national level. 
				Assuming all states have equal populations, the sensor map $H$ is
				$$ 
				H = 
				\begin{bmatrix}
				1 & 0 & 0 & 0 & 0\\
				0 & 1 & 0 & 0 & 0 \\
				0 & 0 & 1 & 0 & 0 \\
				0 & 0 & 0 & 1 & 0 \\
				0 & 0 & 0 & 0 & 1 \\
				\sfrac{1}{3}  & \sfrac{1}{3} & \sfrac{1}{3}  & 0 & 0\\
				0 & 0 & 0 & \sfrac{1}{2} & \sfrac{1}{2} \\
				\sfrac{1}{5}  & \sfrac{1}{5}  & \sfrac{1}{5} & \sfrac{1}{5} &\sfrac{1}{5} \\ 
				\end{bmatrix}.
				$$}
			\label{fig:flu_simple}}
		{\includegraphics[width=0.55\textwidth]{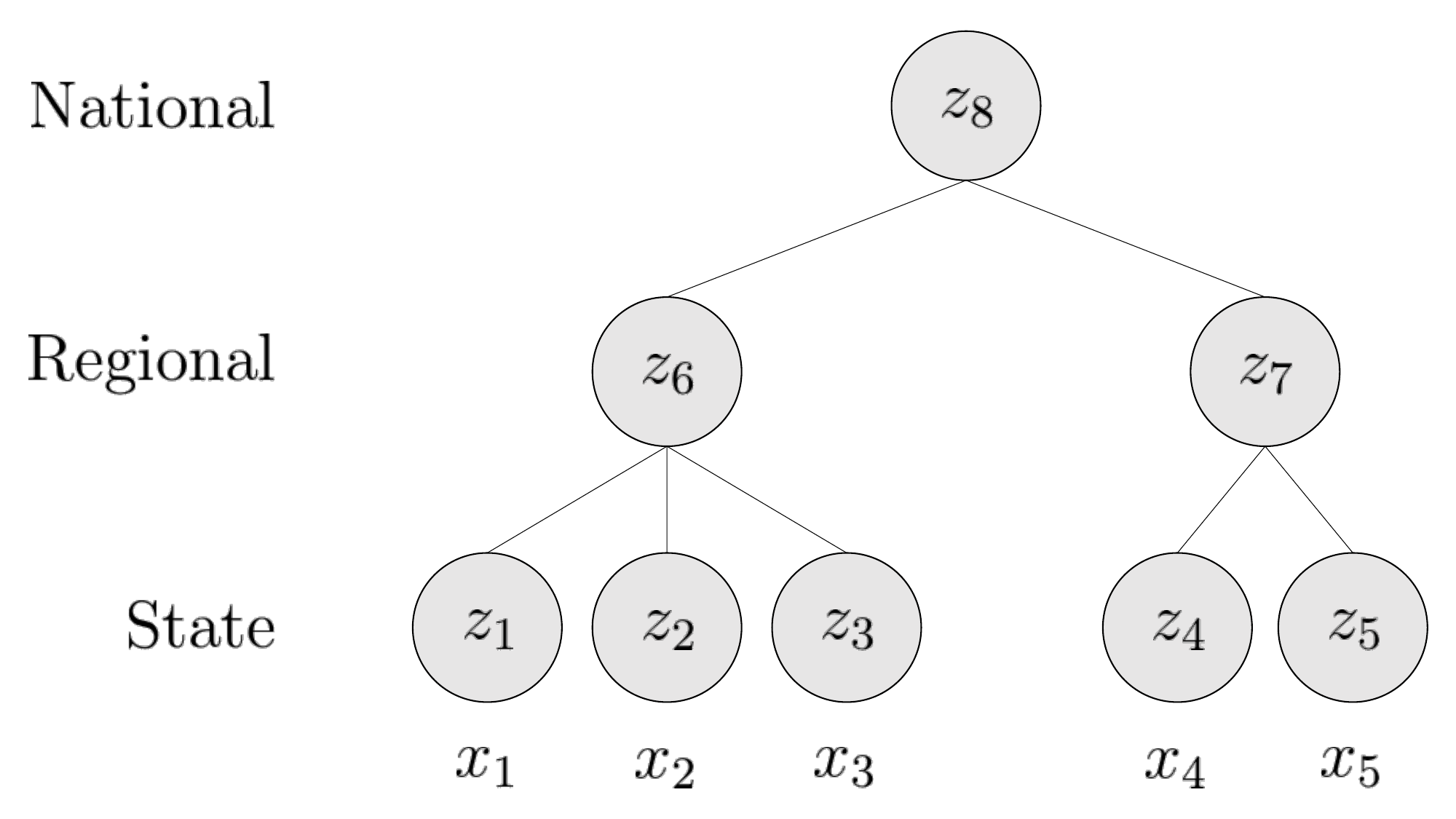}} 
		\vspace{-15pt}
	\end{figure}
	
	\subsection{Interpreting the constraints}
	
	At a high-level, the constraints in \eqref{eq:con_reg} encode information
	about the measurement model \eqref{eq:measure}. They also provide some kind 
	of implicit regularization.  Interestingly, as we will see later in Section
	\ref{sec:flu_demo}, this can still be useful when used in addition to more
	typical (explicit) regularization.  
	
	How can we interpret these constraints?  We give three interpretations, the
	first one specific to the flu forecasting setting, and the next two general.
	
	\paragraph{Flu interpretation.} In the flu nowcasting problem, recall, the map 
	$H$ has rows that sum to 1, and they reflect the geographic level at which
	the corresponding sensors were trained (see Section \ref{sec:flu_setup}).  The
	constraints $H^T b_j=e_j$, $j=1,\ldots,k$ can be seen in this case as a
	mechanism that accounts for the geographical hierachy underlying the sensors.
	As a concrete example, consider the simplified setup in Figure
	\ref{fig:flu_simple}, and $j=3$. The constraint $H^T b_3=e_3$ reads: 
	\begin{align*}
		b_{31} + \sfrac{1}{3} \, b_{36} + \sfrac{1}{5} \, b_{38} &= 0, \\
		b_{32} + \sfrac{1}{3} \, b_{36} + \sfrac{1}{5} \, b_{38} &= 0, \\
		b_{33} + \sfrac{1}{3} \, b_{36} + \sfrac{1}{5} \, b_{38} &= 1, \\
		b_{34} + \sfrac{1}{3} \, b_{37} + \sfrac{1}{5} \, b_{38} &= 0, \\
		b_{35} + \sfrac{1}{3} \, b_{37} + \sfrac{1}{5} \, b_{38} &= 0. 
	\end{align*}
	The third line can be interpreted as follows: an increase of 1 unit in
	sensor $z_3$, $1/3$ units in $z_6$, and $1/5$ units in $z_8$, holding all other
	sensors fixed, should lead to an increase in 1 unit of our prediction for
	$x_3$. This is a natural consequence of the hierarchy in the sensor model
	\eqref{eq:measure}, visualized in Figure \ref{fig:flu_simple}.  The first line
	can be read as: an increase of 1 unit in sensor $z_1$, $1/3$ units in $z_6$, and
	$1/5$ in $z_8$, with all others fixed, should not change our prediction for
	$x_3$. This is also natural, following from the hierachy (i.e., such a change
	must have been propogated by $x_1$).  The other lines are similar. 
	
	\paragraph{Invariance interpretation.} The SF prediction (at time $t+1$) is
	\smash{$\hx_{t+1} = \hB^T z_{t+1}$}.  To denoise (i.e., estimate the mean of)
	the measurement $z_{t+1}$, based on the model \eqref{eq:measure}, we 
	could use \smash{$\hat{z}_{t+1} = H \hx_{t+1}$}.  Given the 
	denoised \smash{$\hat{z}_{t+1}$}, we could then refit our state prediction via
	\smash{$\tilde{x}_{t+1}=\hB^T \hat{z}_{t+1}$}.  But due to the constraint
	\smash{$H^T \hB = I$} (a compact way of expressing \smash{$H^T \hb_j=e_j$}, for
	$j=1,\ldots,k$), it holds that \smash{$\tilde{x}_{t+1} = \hB^T H \hx_{t+1} =
		\hx_{t+1}$}. This is a kind of {\it invariance} property.  In other words, we
	can go from estimating states, to refitting measurements, to refitting states,
	etc., and in this process, our state estimates will not change.  
	
	\paragraph{Generative interpretation.} Assume $t \geq k$, and fix an 
	arbitrary $j = 1,\ldots,k$ as well as $b_j \in \R^k$.  The constraint $H^T b_j = 
	e_j$ implies, by taking an inner product on both sides with $x_i$,
	$i=1,\ldots,k$, 
	$$
	(Hx_i)^T b_j = x_{ij}, \quad i=1,\ldots,k.
	$$
	If we assume $x_i$, $i=1,\ldots,k$ are linearly independent, then the above
	linear equalities are not only implied by $H^T b_j = e_j$, they are actually
	equivalent to it.  Invoking the model \eqref{eq:measure}, we may rewrite
	the constraint $H^T b_j = e_j$ as
	\begin{equation}
		\label{eq:gen_int}
		\E(b_j^T z_i | x_i) = x_{ij}, \quad i=1,\ldots,k.
	\end{equation}
	In the context of problem \eqref{eq:con_reg}, this is a statement about a {\it
		generative} model for the data (as $z_i|x_i$ describes the distribution of
	the covariates conditional on the response). The representation in
	\eqref{eq:gen_int} shows that \eqref{eq:con_reg} constrains the regression
	estimator to have the correct conditional predictions, on average, on the data
	we have already seen $(x_i,z_i)$, $i=1,\ldots,k$.  (Note here we did not have to
	use the first $k$ time points; any past $k$ time points would suffice.) 
	
	\subsection{Modifications and equivalences}
	\label{sec:sf_mod}
	
	In the supplement, we show that two modifications of the basic SF formulation
	also have equivalences in the regression perspective: namely, shrinking the
	empirical covariance in \eqref{eq:emp_cov} towards the identity is equivalent to
	adding a ridge (squared $\ell_2$) penalty to the criterion in
	\eqref{eq:con_reg}; and also, adding a null sensor at each state (one that
	always outputs 0) is equivalent to removing the constraints in
	\eqref{eq:con_reg}.   The latter equivalence here provides indirect but
	fairly compelling evidence that the constraints in the regression formulation
	\eqref{eq:con_reg} play an important role (under the model \eqref{eq:measure}): 
	it says that removing them is equivalent to including meaningless null sensors,
	which intuitively should worsen its predictions.        
	
	\section{Flu nowcasting application}
	\label{sec:flu_demo}
	
	\paragraph{Experimental setup.}  We examine the performance of our methods for
	nowcasting (one-week-ahead prediction of) wILI across 5 flu seasons, from 2013
	to 2018 (total of 140 weeks). Recall the setup described in Section
	\ref{sec:flu_setup}, with $k=51$ states and $d = 308$ measurements. At week
	$t+1$, we derive an estimate \smash{$\hx_{t+1}$} of the current wILI in the 51
	US states, based on sensors $z_{t+1}$ (each sensor being the output of an
	algorithm trained to predict wILI at a different geographic resolution from a
	given data source), and past wILI and sensor data. We consider 7 methods for
	computing the nowcast \smash{$\hx_{t+1}$}: (i) SF, or equivalently, constrained
	regression \eqref{eq:con_reg}; (ii) SF as in \eqref{eq:con_reg}, but with an
	additional ridge (squared $\ell_2$) penalty (equivalently, SF with covariance
	shrinkage); (iii) SF as in \eqref{eq:con_reg}, but with an additional lasso
	($\ell_1$) penalty; (iv/v) regression as in \eqref{eq:con_reg}, but without
	constraints, and using a ridge/lasso penalty; (vi) random forests (RF)
	\citep{breiman2001random}, trained on all of the sensors; (vii) RF, but trained 
	on all of the underlying data sources used to fit the sensors. 
	
	At prediction week $t+1$, we use the last 3 years (weeks $t-155$ through $t$) as
	the training set for all 7 methods. We do not implement unpenalized 
	regression (as in \eqref{eq:con_reg}, but without constraints), as it is not
	well-defined (156 observations and 308 covariates).\footnote{SF is still
		well-defined, due of the constraint in \eqref{eq:con_reg}: a nonunique
		solution only occurs when the (random) null space of the covariate matrix has
		a nontrivial intersection with the null space of $H^T$, which essentially
		never happens.}  All ridge and lasso tuning parameters are chosen by
	optimizing one-week-ahead prediction error over the latest 10 weeks of data
	(akin to cross-validation, but for a time series context like ours). Python code
	for this nowcasting experiment is available at
	\url{http://github.com/mariajahja/kf-sf-flu-nowcasting}.     
	
	\begin{figure}[t!]
		\centering
		\includegraphics[width=\textwidth]{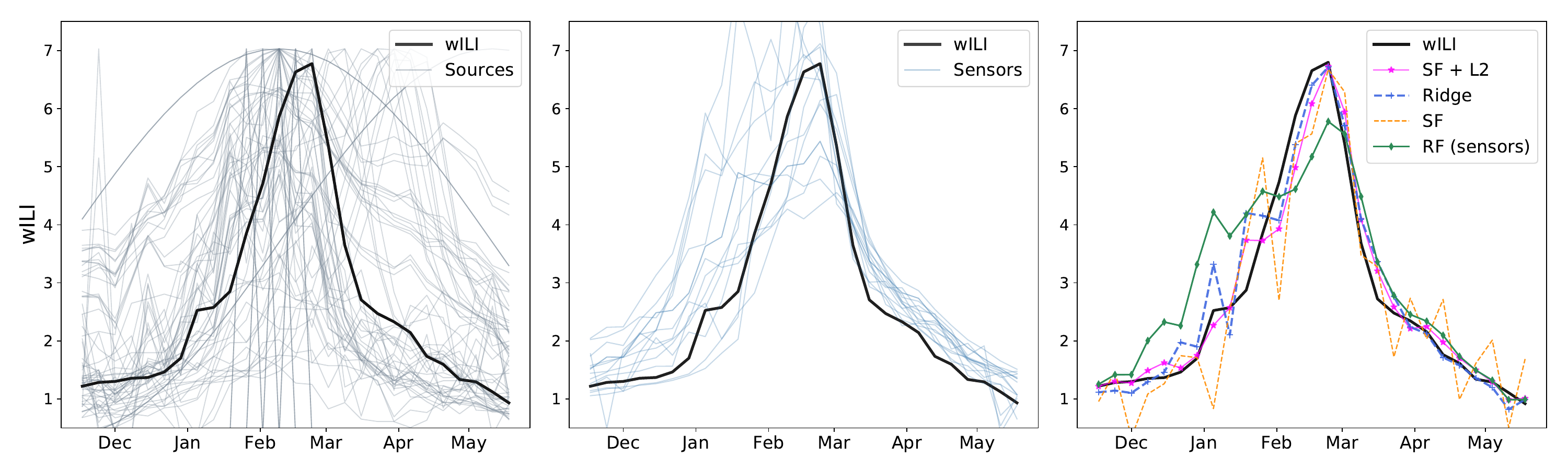} \\
		\smallskip\smallskip
		\includegraphics[width=0.95\textwidth]{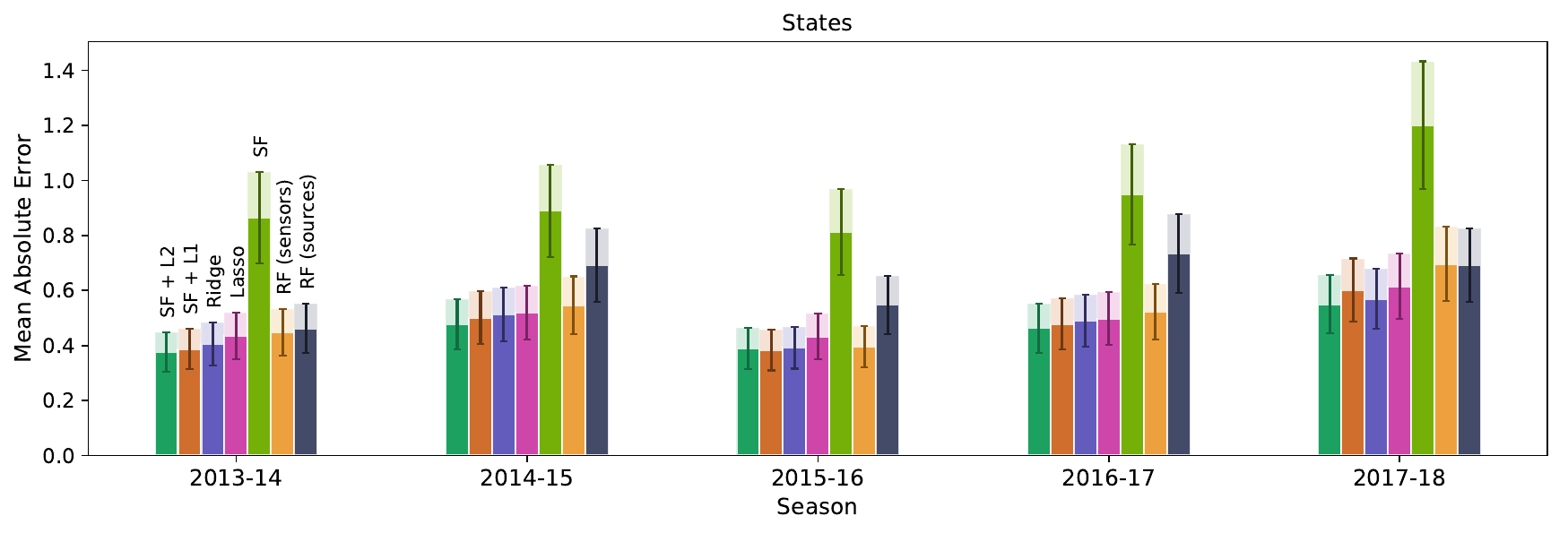}
		\caption{\it\small Top row, from left to right: data sources, sensors, and
			nowcasts are compared to the underlying wILI values for Pennsylvania during
			flu season 2017-18. For visualization purposes, the sources are scaled to fit 
			the range of wILI. On the rightmost plot, we display nowcasts using select
			methods. Bottom row: MAEs (full colors) and MADs (light colors) of nowcasts 
			over 5 flu seasons from 2013-14 to 2017-18.} 
		\label{fig:flu_results}
	\end{figure}
	
	\paragraph{Missing data.} Unfortunately, sensors are observed at not only
	varying geographic resolutions, but also varying temporal resolutions (since 
	their underlying data sources are), and missing values occur.  In our
	experiments, we choose to compute predictions using the regression perspective,
	and apply a simple mean imputation approach (using only past sensor data),
	before fitting all models.  
	
	\paragraph{Nowcasting results.} The bottom row of Figure \ref{fig:flu_results}
	displays the mean absolute errors (MAEs) from one-week-ahead predictions by the
	7 methods considered, averaged over the 51 US states, for each of the 5
	seasons. Also displayed are the mean absolute deviations (MADs), in light  
	colors. We see that SF with ridge regularization is generally the most accurate
	over the 5 seasons, SF with lasso regularization is a close second, and SF
	without any regularization is the worst.  Thus, clearly, explicit regularization
	helps. Importantly, we also see that the constraints in the regression problem 
	\eqref{eq:con_reg} (which come from its connection to SF) play a key role: in
	each season, SF with ridge regularization outperforms ridge regression, and SF
	with lasso regularization outperforms the lasso.  Therefore, the constraints
	provide additional (beneficial) implicit regularization.   
	
	RF trained on sensors performs somewhat competitively.  RF trained on sources is
	more variable (in some seasons, much worse than RF on sensors).  This 
	observation indicates that training the sensors is an important step for
	nowcasting accuracy, as this can be seen as a form of denoising, and
	suggests a view of all the methods we consider here (except RF on sources) as
	prediction assimilators (rather than data assimilators).  Finally, the top row
	Figure \ref{fig:flu_results} visualizes the nowcasts for Pennsylvania in the
	2017-18 season.  We can see that SF, RF (on sensors), and even ridge regression
	are noticeably more volatile than SF with ridge regularization.  
	
	\section{Discussion and extensions}
	\label{sec:end}
	
	In this paper, we studied connections between the Kalman filter, sensor fusion,
	and regression. We derived equivalences between the first two and latter two, 
	discussed the general implications of our results, and studied the application
	of our work to nowcasting the weekly influenza levels in the US. We conclude
	with some ideas for extending the constrained regression formulation
	\eqref{eq:con_reg} of SF.    
	
	
	\paragraph{Sensor selection.} The problem of selecting a small number of relevant
	sensors (on which to perform sensor fusion) among a possibly large number, which 
	we can call {\it sensor selection}, is quite a difficult problem.  Beyond this,
	measurement selection in the Kalman filter is a generally difficult problem. As
	far as we know, this is an active and relatively open area of research. 
	On the other hand, in regression, variable selection is extremely well-studied,
	and $\ell_1$ regularization (among many other tools) is now very well-developed
	(see, e.g., \citet{hastie2009elements,hastie2015statistical}).  Starting from the
	regression formulation for SF in \eqref{eq:con_reg}, it would be natural
	to add to the criterion an $\ell_1$ or {\it lasso} penalty
	\citep{tibshirani1996regression} to select relevant sensors, 
	\begin{equation}
		\label{eq:con_lasso}
		\begin{alignedat}{2}
			&\minimize_{b_j \in \R^d} && \frac{1}{t} \sum_{i=1}^t (x_{ij} - b_j^T z_i)^2 +
			\lambda_j \|b_j\|_1 \\  
			&\st \quad && H^T b_j = e_j,
		\end{alignedat}
	\end{equation}
	where \smash{$\|b_j\|_1=\sum_{\ell=1}^k|b_{j\ell}|$}, $j=1,\ldots,k$.  It is
	not clear (nor likely) that \eqref{eq:con_lasso} has an equivalent SF
	formulation, but the exact equivalence when $\lambda_j=0$ suggests that   
	\eqref{eq:con_lasso} could be a reasonable tool for sensor selection.  (Indeed, 
	without even considering its sensor selection capabilities, this performed
	respectably for predictive purposes in the experiments in Section
	\ref{sec:flu_demo}.)  Further, we can perform a kind of process model selection
	with \eqref{eq:con_lasso} by augmenting our measurement vector with multiple 
	candidate process models, and penalizing only their coefficients. An example is
	given in the supplement.  
	
	\paragraph{Joint sensor learning.} In the flu nowcasting problem, recall, the
	sensors are outputs of predictive models, each trained individually to predict
	wILI from a particular data source (flu proxy).  Denote by $u_i \in \R^d$,
	$i=1,\ldots,t$ the data sources at times 1 through $t$. Instead of learning the
	sensors (predictive transformations of these sources) individually, we could
	learn them jointly, by extending \eqref{eq:con_reg} into: 
	\begin{equation}
		\label{eq:con_joint}
		\begin{alignedat}{2}
			&\minimize_{\substack{b_j \in \R^d, \, j=1,\ldots,d \\ f \in \mathcal{F}}} \quad
			&& \frac{1}{t} \sum_{j=1}^ d \sum_{i=1}^t \big(x_{ij} -  b_j^T f(u_i) \big)^2 
			+ \lambda P(f) \\   
			&\;\; \st \quad && H^T b_j = e_j, \quad j=1,\ldots,k.
		\end{alignedat}
	\end{equation}
	Here $\mathcal{F}$ is a space of functions from $\R^d$ to $\R^d$ (e.g., diagonal
	linear maps) and $P$ is a penalty to be specified by the modeler (e.g., the
	Frobenius norm in the linear map case). The key in \eqref{eq:con_joint} is that
	we are simultaneously learning the sensors and assimilating them.     
	
	\paragraph{Gradient boosting.}  Solving \eqref{eq:con_joint} is computationally 
	difficult (even in the simple linear map case, it is nonconvex).  An
	alternative that is more tractable is to proceed iteratively, in a manner
	inspired by {\it gradient boosting} \citep{friedman2001greedy}.  For each
	$j=1,\ldots,d$, let $A_j$ be an algorithm (``base learner'') that we use to
	fit sensor $j$ from data source $j$.  Write $y_i=Hx_i$, $i=1,\ldots,t$, and let
	$\eta > 0$ be a small fixed learning rate.  To make a prediction at time $t+1$,  
	we initialize \smash{$x_i^{(0)}=0$}, $i=1,\ldots,t+1$ (or initialize at the fits
	from the usual linear SF), and repeat, for boosting iterations $b=1,\ldots,B$:
	\begin{itemize}
		\item For $j=1,\ldots,d$:
		\begin{itemize}
			\item Let {$y^{(b-1)}_{ij}=(Hx^{(b-1)})_{ij}$}, for $i=1,\ldots,t$. 
			\item Run $A_j$ with responses {$\{y_{ij}-y^{(b-1)}_{ij}\}_{i=1}^t$} and
			covariates {$\{u_{ij}\}_{i=1}^t$}, to produce 
			{$\bar{f}^{(b)}_j$}.   
			\item Define intermediate sensors
			{$z^{(b)}_{ij}=\bar{f}^{(b)}_j(u_{ij})$}, for $i=1,\ldots,t+1$. 
		\end{itemize}
		\item For $j=1,\ldots,k$:
		\begin{itemize} 
			\item Run SF as in \eqref{eq:con_reg} (possibly with regularization) with
			responses {$\{x_{ij}-x^{(b-1)}_{ij}\}_{i=1}^t$} and 
			covariates {$\{z^{(b)}\}_{i=1}^t$}, to produce {$\hb_j$}.
			\item Define intermediate state fits {$\bar{x}^{(b)}_{ij}=\hb_j^T 
				z^{(b)}_i$}, for $i=1,\ldots,t+1$.  
			\item Update total state fits {$x^{(b)}_{ij} = x^{(b-1)}_{ij} +
				\eta \bar{x}^{(b)}_{ij}$}, for $i=1,\ldots,t+1$.    
		\end{itemize}
	\end{itemize}
	We return at the end our final prediction \smash{$\hx_{t+1}=x^{(B)}_{t+1}$}.
	It would be interesting to pursue this approach in detail, and study the extent
	to which it can improve on the usual linear SF.
	
	
	\paragraph{Acknowledgments.} We thank Logan Brooks for several helpful
	conversations and brainstorming sessions. MJ was supported by NSF Graduate
	Research Fellowship No.\ DGE-1745016.  RR and RJT were supported by DTRA
	Contract No.\ HDTRA1-18-C-0008.

 \bibliographystyle{plainnat}
 \bibliography{arxiv}
 
\appendix

\newpage
\section{Proofs and additional details}
\renewcommand\thefigure{A.\arabic{figure}}    
\renewcommand\thesection{A.\arabic{section}}
\renewcommand\theequation{A.\arabic{equation}}     
\renewcommand\thetheorem{A.\arabic{theorem}}     
\renewcommand\thelemma{A.\arabic{lemma}}     
\renewcommand\theremark{A.\arabic{remark}}    
\setcounter{section}{0}
\section{Proof of Theorem \ref{thm:kf_sf}}

We can write the sensor fusion update as 
\begin{align*}
\tilde{P}_{t+1} &= (\tilde{H}^T \tilde{R}_{t+1}^{-1} \tilde{H})^{-1} \\ 
\hx_{t+1} &= \tilde{P}_{t+1} \tilde{H}^T \tilde{R}_{t+1}^{-1} \tilde{z}_{t+1},
\end{align*}
where
$$
\tilde{P}_{t+1} = (H^T R^{-1} H + \bP_{t+1}^{-1})^{-1}.
$$
By the Woodbury matrix identity, $(A+UCV^{-1}) = A^{-1} - A^{-1} U
(C^{-1} + VA^{-1}U)^{-1} V A^{-1}$, with \smash{$A=\bP_{t+1}^{-1}$} in our 
case, we get
\begin{align}
\nonumber
\tilde{P}_{t+1} &= \bP_{t+1} - \bP_{t+1} H^T (R + H \bP_{t+1} H^T)^{-1} H 
\bP_{t+1} \\
\nonumber
&= (I - \bP_{t+1} H^T (R + H \bP_{t+1} H^T)^{-1} H) \bP_{t+1} \\
\label{eq:p_equality}
&= (I - K_{t+1} H) \bP_{t+1},
\end{align}
where recall, the Kalman gain $K_{t+1}$ is defined in \eqref{eq:k_update}.   

Now let us we rewrite the Kalman gain as
\begin{align*}
K_{t+1} &= \bP_{t+1} H^T (R + H \bP_{t+1} H^T)^{-1} \\
&= \bP_{t+1} H^T R^{-1}  (I + H \bP_{t+1} H^T R^{-1})^{-1},
\end{align*}
so that 
$$
K_{t+1} (I + H \bP_{t+1} H^T R^{-1}) = \bP_{t+1} H^T R^{-1}, 
$$
and after rearranging, 
\begin{equation}
\label{eq:k_equality}
K_{t+1} = (I - K_{t+1} H) \bP_{t+1} H^T R^{-1}.
\end{equation}

Putting \eqref{eq:p_equality} and \eqref{eq:k_equality} together, we get   
\begin{align*}
\tilde{P}_{t+1} \tilde{H}^T \tilde{R}_{t+1}^{-1} \tilde{z}_{t+1} 
&= (I - K_{t+1} H) \bP_{t+1} (H^T R^{-1} z_{t+1} + \bP_{t+1}^{-1}
\bx_{t+1}) \\  
&= (I - K_{t+1} H) \bP_{t+1} H^T R^{-1} z_{t+1} + (I - K_{t+1} H) \bx_{t+1} \\
&= K_{t+1} z_{t+1} + (I - K_{t+1} H) \bx_{t+1} \\
&= \bx_{t+1} + K_{t+1} (z_{t+1} - H \bx_{t+1}),
\end{align*}
which is exactly the Kalman filter prediction, completing the proof.

\section{Derivation of \eqref{eq:kf_sf_nonlinear}}

We first make the EKF estimate precise.  Let 
\begin{align}
\label{eq:f_jacobian}
F_{t+1} &= Df(\hx_t), \\
\label{eq:h_jacobian}
H_{t+1} &= Dh(\bx_{t+1}),
\end{align}
and define 
\begin{align}
\label{eq:ekf_predict}
\bx_{t+1} &= F_{t+1} \hx_t, \\
\label{eq:ekf_update}
\hx_{t+1} &= \bx_{t+1} + K_{t+1} \big(z_{t+1} - h(\bx_{t+1})\big), 
\end{align}
where $K_{t+1} \in \R^{k \times d}$ is defined via 
\begin{align}
\label{eq:ep_predict}
\bP_{t+1} &= F_{t+1} P_t F_{t+1}^T + Q, \\
\label{eq:ek_update}
K_{t+1} &= \bP_{t+1} H_{t+1}^T (H_{t+1} \bP_{t+1} H_{t+1}^T + R)^{-1}, \\ 
\label{eq:ep_update}
P_{t+1} &= (I-K_{t+1} H_{t+1}) \bP_{t+1},
\end{align}
Note that \eqref{eq:ep_predict}--\eqref{eq:ep_update} are exactly the same as 
\eqref{eq:p_predict}--\eqref{eq:p_update}, with $F_{t+1},H_{t+1}$ replacing
$F,H$, respectively.  Moreover, \eqref{eq:ekf_predict}, \eqref{eq:ekf_update}
are {\it nearly} the same as \eqref{eq:kf_predict}, \eqref{eq:kf_update}, with
again $F_{t+1},H_{t+1}$ replacing $F,H$, except that the residual in
\eqref{eq:ekf_update} is \smash{$z_{t+1} - h(\bx_{t+1})$}, and not
\smash{$z_{t+1} - H_{t+1}\bx_{t+1}$}, as would be analogous from
\eqref{eq:kf_update}.    

Next, we make what we called the extended SF (ESF) estimate precise. 
Let \smash{$\tilde{z}_{t+1} \in \R^{d+k}$} be as in
\eqref{eq:tz_nonlinear},  let \smash{$\tilde{H}_{t+1} \in \R^{(d+k) \times k}$}
be the rowwise concatentation of $H_{t+1}$ and $I \in \R^{k \times k}$, and 
\smash{$\tilde{R}_{t+1}$} be as in \eqref{eq:block_cov}.  Here,
\smash{$F_{t+1},H_{t+1},\bP_{t+1}$} are as defined in \eqref{eq:f_jacobian},
\eqref{eq:h_jacobian}, \eqref{eq:ep_predict}, respectively.  The ESF
estimate is 
\begin{equation}
\label{eq:esf}
\hx_{t+1} = (\tilde{H}^T \tilde{R}_{t+1}^{-1} \tilde{H})^{-1} \tilde{H}^T
\tilde{R}_{t+1}^{-1} \tilde{z}_{t+1}.
\end{equation}

To see that \eqref{eq:esf} and \eqref{eq:ekf_update} are equal, note that by 
following the proof of Theorem \ref{thm:kf_sf} directly, with $F_{t+1},H_{t+1}$
in place of $F,H$, we get
$$
(\tilde{H}_{t+1}^T \tilde{R}_{t+1}^{-1} \tilde{H}_{t+1})^{-1} \tilde{H}_{t+1}^T  
\tilde{R}_{t+1}^{-1} \tilde{z}_{t+1} = 
\bx_{t+1} + K_{t+1} (z_{t+1} - H_{t+1} \bx_{t+1}).
$$
Adding and subtracting \smash{$K_{t+1} h(\bx_{t+1})$} to the right-hand side
gives
\begin{align*}
(\tilde{H}_{t+1}^T &\tilde{R}_{t+1}^{-1} \tilde{H}_{t+1})^{-1} \tilde{H}_{t+1}^T  
\tilde{R}_{t+1}^{-1} (z_{t+1}, \bx_{t+1}) \\
&= \bx_{t+1} + K_{t+1} \big(z_{t+1} - h(\bx_{t+1})\big) + 
K_{t+1}(h(\bx_{t+1} - H_{t+1} \bx_{t+1}) \\
&= \bx_{t+1} + K_{t+1} \big(z_{t+1} - h(\bx_{t+1})\big) + 
(I - K_{t+1} H_{t+1}) \bP_{t+1} H_{t+1}^T R^{-1} (h(\bx_{t+1} - H_{t+1}
\bx_{t+1}) \\ 
&= \bx_{t+1} + K_{t+1} \big(z_{t+1} - h(\bx_{t+1})\big) + 
\tilde{P}_{t+1} H_{t+1}^T R^{-1} (h(\bx_{t+1} - H_{t+1} \bx_{t+1}),
\end{align*}
where in the second line we used \eqref{eq:k_equality}, and in the third we
used \eqref{eq:p_equality}.  Rearranging gives
$$
(\tilde{H}_{t+1}^T \tilde{R}_{t+1}^{-1} \tilde{H}_{t+1})^{-1} \tilde{H}_{t+1}^T  
\tilde{R}_{t+1}^{-1} \big(z_{t+1} + H_{t+1} \bx_{t+1} - h(\bx_{t+1}), \,
\bx_{t+1}\big) = \bx_{t+1} + K_{t+1} \big(z_{t+1} - h(\bx_{t+1})\big),
$$
which is precisely the desired conclusion, in \eqref{eq:kf_sf_nonlinear}.

\section{Proof of Theorem \ref{thm:sf_reg}}

Let us denote $X \in \R^{t \times k}$ and $Z \in \R^{t \times d}$ the matrices of
states and sensors, respectively, for the first $t$ time points.  That is, $X$  
has rows $x_i \in \R^k$, $i=1,\ldots,t$ and $Z$ has rows $z_i \in \R^d$,
$i=1,\ldots,t$. Fix any $j=1,\ldots,k$. Let \smash{$\hat{a}_j \in \R^d$} be
the $j$th column of \smash{$\hR_{t+1}^{-1} H (H^T \hR_{t+1}^{-1} H)^{-1}$}, and
let \smash{$\hb_j \in \R^d$} be the solution of \eqref{eq:con_reg},
equivalently, the solution of  
\begin{equation}
\label{eq:con_reg_mat}
\begin{alignedat}{2}
&\minimize_{b_j \in \R^d} && \|X_j - Z b_j\|_2^2 \\ 
&\st \quad && H^T b_j = e_j,
\end{alignedat}
\end{equation}
where $X_j$ denotes the $j$th column of $X$.  We will show that
\smash{$\hat{a}_j=\hb_j$}.  

The Lagrangian of problem \eqref{eq:con_reg_mat} is  
$$
L(b_j,u_j) = \|X_j - Z b_j\|_2^2 + u_j^T (H^T b_j - e_j),
$$
for a dual variable (Lagrange multiplier) $u_j \in \R^k$.  Taking the gradient 
of the Lagrangian and setting it equal to zero at an optimal pair
\smash{$(\hb_j,\hat{u}_j)$} gives 
$$
0 = Z^T(Z \hb_j - X_j) + H \hat{u}_j,
$$
and rearranging gives
\begin{equation}
\label{eq:stat_cond}
\hb_j = (Z^T Z)^{-1} (Z^T X_j - H \hat{u}_j).
\end{equation}
The dual solution \smash{$\hat{u}_j$} can be determined by plugging
\eqref{eq:stat_cond} into the equality constraint \smash{$H^T \hb_j=e_j$},
but for our purposes, the explicit dual solution is unimportant. 

We will now show that \smash{$\hat{b}_j = \hR_{t+1}^{-1} H \hat\beta_j$} for
some \smash{$\hat\beta_j \in \R^k$}.  Write
\begin{align*}
\hR_{t+1} &= \frac{1}{t} (Z - XH^T)^T (Z-XH^T) + (1-\alpha) I \\ 
&= \frac{1}{t} (Z^T Z - HX^T Z - Z^T XH^T + HX^T X H^T). 
\end{align*}
Then 
\begin{align*}
\hR_{t+1} \hb_j &= \frac{1}{t}(Z^T Z \hb_j - HX^T Z \hb_j - Z^T XH^T \hb_j +
HX^T X H^T \hb_j) \\
&= \frac{1}{t}(Z^T X_j - H\hat{u}_j - HX^T Z \hb_j - Z^T X_j + HX^T X_j) \\
&= H\underbrace{\Big(\frac{X^T X_j - \hat{u}_j - X^T Z
		\hb_j}{t}\Big)}_{\hat\beta_j}, 
\end{align*}
as desired, where in the second line we have used \eqref{eq:stat_cond} and the
constraint \smash{$H^T \hb_j = e_j$}. 

Observe that \smash{$\hat{a}_j = \hR_{t+1}^{-1} H \hat\alpha_j$} for some 
\smash{$\hat\alpha_j \in \R^k$}, in particular, for \smash{$\hat\alpha_j$} 
defined to be the $j$th column of \smash{$(H^T \hR_{t+1}^{-1} H)^{-1}$}.
Further,  
$$
e_j = H^T \hat{a}_j = H^T \hb_j
$$
the constraint on \smash{$\hat{a}_j$} holding by direct verification, and the
constraint on \smash{$\hb_j$} holding by construction in
\eqref{eq:con_reg_mat}.  That is,
$$
H^T \hR_{t+1}^{-1} H \hat\alpha_j = H^T \hR_{t+1}^{-1} H \hat\beta_j, 
$$
and since \smash{$H^T \hR_{t+1}^{-1} H$} is invertible, this leads to 
\smash{$\hat\alpha_j = \hat\beta_j$}, and finally \smash{$\hat{a}_j=\hb_j$},
completing the proof.  

\section{Further SF-regression equivalences} 

\subsection{More regularization: covariance shrinkage}

Covariance shrinkage---which broadly refers to the technique of adding a 
well-conditioned matrix to a covariance estimate to provide stability
and regularity---is widely used and well-studied in modern multivariate 
statistics, data mining, and machine learning.  As such, it would be natural to
replace the empirical covariance matrix estimate \eqref{eq:emp_cov} for the
measurement noise covariance by
\begin{equation}
\label{eq:shrink_cov}
\hR_{t+1} = \frac{\alpha}{t} \sum_{i=1}^t (z_i - Hx_i)(z_i - Hx_i)^T +
(1-\alpha) I, 
\end{equation}
for a parameter $\alpha \in [0,1]$.  For sensor fusion in the flu nowcasting
problem, this is considered (in some form) in \citet{farrow2016modeling}, and 
leads to significant improvements in nowcasting accuracy. 

Our next result shows that when we use shrinkage as in \eqref{eq:shrink_cov} to
estimate the measurement noise covariance in SF, this is equivalent to adding a
ridge penalty in the regression formulation. 

\begin{corollary}
	\label{cor:sf_ridge}
	Let \smash{$\hR_{t+1}$} be as in \eqref{eq:shrink_cov}, for some value 
	$\alpha \in [0,1]$. Consider the SF prediction at time $t+1$, with
	\smash{$\hR_{t+1}$} in place of $R$, denoted \smash{$\hx_{t+1} = \hB^T
		z_{t+1}$}.  Then each column of \smash{$\hB$}, denoted \smash{$\hb_j \in
		\R^d$}, $j=1,\ldots,k$, solves    
	\begin{alignat*}{2}
	&\minimize_{b_j \in \R^d} && \frac{1}{t} \sum_{i=1}^t (x_{ij} - b_j^T z_i)^2 + 
	\frac{(1-\alpha)}{\alpha} \| b_j \|_2^2 \\
	&\st \quad && H^T b_j = e_j.
	\end{alignat*}
\end{corollary}

\begin{proof}
	As before, let $X \in \R^{t \times k}$ and $Z \in \R^{t \times d}$ denote the
	matrix of states and sensors, respectively, over the first $t$ time points.  We
	can write \smash{$\hR_{t+1}$} in \eqref{eq:shrink_cov}
	$$
	\frac{\alpha}{t} (Z - XH^T)^T (Z-XH^T) + (1-\alpha) I = 
	\frac{1}{t} (\tilde{Z} - \tilde{X} H^T)^T (\tilde{Z} - \tilde{X} H^T),
	$$
	where \smash{$\tilde{Z} \in \R^{(t+d) \times d}$} is the rowwise concatenation
	of \smash{$\sqrt{\alpha/t} Z$} and \smash{$\sqrt{1-\alpha/t} I$}, and  
	\smash{$\tilde{X} \in \R^{(t+k) \times k}$} is the rowwise concatenation of
	\smash{$\sqrt{\alpha/t} X$} and $0 \in \R^{k \times k}$ (the matrix of all 0s).
	Applying Theorem \ref{thm:sf_reg} to \smash{$\tilde{X}, \tilde{Z}$}, expanding
	the criterion in the regression problem, and then multiplying the criterion by 
	$1/\alpha$, gives the result.
\end{proof}

\subsection{Less regularization: zero padding}

In the opposite direction, we now show that we can modify SF and obtain an 
equivalent regression formulation with less regularization, specifically,
without constraints. 

\begin{corollary}
	\label{cor:sf_uncon}
	At each $t=1,2,3,\ldots$, suppose we augment our measurement vector by 
	introducing $k$ measurements that are identically zero, denoted
	\smash{$\tilde{z}_t = (z_t, 0) \in \R^{d+k}$}.  Suppose that we augment our
	measurement map accordingly, defining \smash{$\tilde{H} \in \R^{(d+k) \times
			k}$} to be the rowwise concatention of $H$ and the identity $I \in
	\R^{k \times k}$. Consider running SF on this augmented system, using the
	empirical covariance to estimate $R$, and let \smash{$\hx_{t+1}=\hB^T z_{t+1}$} 
	denote the SF prediction at time $t+1$.  Then each column of \smash{$\hB$},
	denoted \smash{$\hb_j \in \R^d$}, $j=1,\ldots,k$, solves
	$$
	\minimize_{b_j \in \R^d} \; \sum_{i=1}^t (x_{ij} - b_j^T z_i)^2.
	$$
\end{corollary}

\begin{proof}
	Applying Theorem \ref{thm:sf_reg} to the augmented system gives the equivalent
	regression problem
	\begin{alignat*}{2}
	&\minimize_{b_j \in \R^d, \, a_j \in \R^k} &&
	\sum_{i=1}^t (x_{ij} - b_j^T z_i - a_j^T 0)^2 \\
	&\st \quad && H^T b_j + a_j = e_j.
	\end{alignat*}
	The constraint is satisfied with $a_j=e_j-H^T b_j$. But $a_j$ has 
	no effect on the criterion, so the constraint can be removed.
\end{proof}

\begin{remark}
	The analogous equivalence holds for covariance shrinkage and ridge regression. 
	That is, in Corollary \ref{cor:sf_uncon}, if instead of the empirical
	covariance, we use $\alpha$ times the empirical covariance plus $(1-\alpha) I$, 
	then SF on the augmented system is equivalent to unconstrained ridge, at tuning
	parameter $(1-\alpha)/\alpha$. 
\end{remark}

\section{Example of process model selection}

Here we give a simple empirical example of process model selection using the 
regression formulation of SF.  We initialized $x_0=1$, and generated data 
according to 
\begin{align*}
x_t &= 0.5 x_{t-1} + 0.05 \sin(0.126 t) + \delta_t, \\ 
z_t &= H x_t + \epsilon_t,
\end{align*}
for $t=1,\ldots,200$.  Here $H \in \R^{4 \times 1}$ is simply the column vector
of all 1s, and the noise is drawn as $\delta_t \sim N(0,0.01)$, $\epsilon_t
\sim N(0,I)$, independently, over $t=1,\ldots,150$.  

The prediction setup is as follows.  At each time $t+1$, when making a
prediction of $x_{t+1}$, we observe all past states $x_i$, $i=1,\ldots,t$ and
all measurements $z_i$, $i=1,\ldots,t+1$.  We fit 5 different candidate process
models to past state data: 
\begin{enumerate}
	\item linear autoregression;
	\item quadratic autoregression;
	\item spline regression on time;
	\item sine regression on time;
	\item cosine regression on time.
\end{enumerate}
To be clear, models 1 and 2 regress $x_i$ on $x_{i-1}$ and $x_{i-1}^2$,
respectively, over $i=1,\ldots,t$.  Models 3--5 regress $x_i$ on a spline, sine,
and cosine transformation of $i$, respectively, over $i=1,\ldots,t$. 
The sine and cosine transformations are given the true frequency.   The spline  
is a cubic smoothing spline (with a knot at every data point) and its
tuning parameter is chosen by cross-validation (using only the past data).  
After being fit, we use each of the candidate process models 1--5 to make a
prediction of $x_{t+1}$, given $z_{t+1}$.  We take this as its ouput.

For $t=151,\ldots,200$, we define $\tilde{z}_t \in \R^9$ to be the measurement
vector $z_t \in \R^4$ augmented with the outputs of the 5 candidate process
models as described above (the burn-in period of 150 time points ensures that 
the candidate process models have enough training data to make reasonable 
predictions).  Figure \ref{fig:pm_select} shows the outputs from these models
over the last 50 time points.  

\begin{figure}[htb]
	\includegraphics[width=0.75\textwidth]{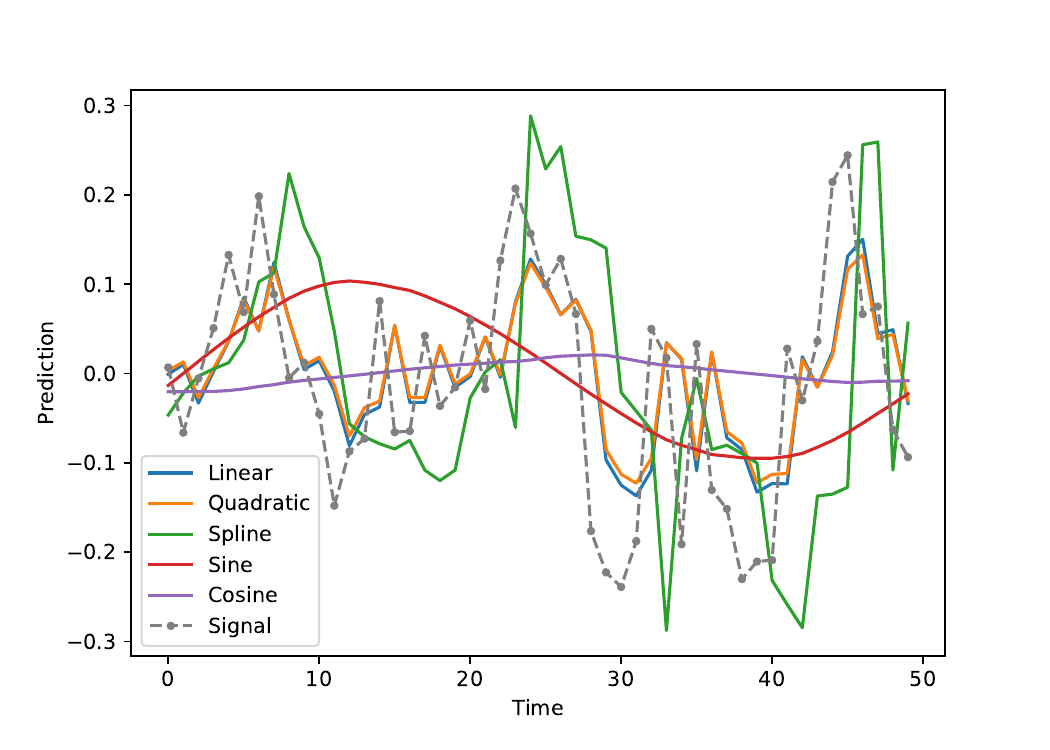}
	\caption{\it\small Simple process model selection example: outputs from 5
		candidate process models, over the last 50 time points.}
	\label{fig:pm_select}
\end{figure}

Finally, in the last 50 time points, to get an assimilated prediction of
\smash{$\hx_{t+1}$} at each time $t+1$, we solve the constrained  regression 
problem with a lasso penalty \eqref{eq:con_lasso}, using 
cross-validation to select $\lambda$ (again, using only past data).
Further, we penalize only the coefficients of the candidate process models
(not the pure measurements).  Table \ref{tab:pm_select} shows the median of
the coefficients over the last 50 time points (in this table, the
coefficients for the pure measurement sensors are aggregated as one).  We see
that the lasso tends to select the linear and sine sensors, as expected (because 
these two make up the true dynamical model), and places a small weight on the
spline sensor (which is flexible, and can mimic the contribution of the sine
sensor).   

\begin{table}[htb]
	\begin{tabular}{lllllll}
		& Linear & Quadratic & Spline & Sine & Cosine & Measurements \\
		\begin{tabular}[c]{@{}l@{}}Median \\ Coefficient\end{tabular} 
		& 0.643 & 0.000 & 0.094 & 0.189 & 0.000 & 0.0175
	\end{tabular}
	\caption{\it\small Simple process model selection example: median regression
		coefficients for the sensors, over the last 50 time points.} 
	\label{tab:pm_select}
\end{table}


\end{document}